\documentclass[]{article} 
\usepackage{mathtools}
\usepackage{float}
\usepackage[utf8]{inputenc}
\usepackage{verbatim}
\usepackage{tikz}
\usepackage{enumerate}
\usepackage{amsmath, amssymb, amsfonts, amsthm}
\usepackage[linesnumbered, boxed, lined, ruled,vlined]{algorithm2e}

\newtheorem{theorem}{Theorem}
\theoremstyle{definition}

\newtheorem{lemma}[theorem]{Lemma}

\newtheorem{observation}[theorem]{Observation}

\renewcommand{\P}{\mathcal{P}}
\newcommand{\NP}{\mathcal{NP}}
\newcommand{\Z}{\mathbb{Z}}
\newcommand{\R}{\mathbb{R}}
\renewcommand{\O}{\mathcal{O}}

\title{Algorithms and Complexity for the Almost Equal Maximum Flow Problem}

\author{R. Haese \and T. Heller\footnote{Corresponding author, \texttt{<till.heller@itwm.fraunhofer.de>}.} \and S.O. Krumke}

\begin{document}
	\maketitle            
	\begin{abstract}
		In the Equal Maximum Flow Problem (EMFP), we aim for a maximum flow where we require the same flow value on all edges in some given subsets of the edge set. In this paper, we study the closely related Almost Equal Maximum Flow Problems (AEMFP) where the flow values on edges of one homologous edge set differ at most by the valuation of a so called deviation function~$\Delta$. We prove that the integer almost equal maximum flow problem (integer AEMFP) is in general $\mathcal{NP}$-complete, and show that even the problem of finding a fractional maximum flow in the case of convex deviation functions is also $\mathcal{NP}$-complete. This is in contrast to the EMFP, which is polynomial time solvable in the fractional case. We provide inapproximability results for the integral AEMFP. For the integer AEMFP we state a polynomial algorithm for the constant deviation and concave case for a fixed number of homologous sets.
	\end{abstract}
	
\section{Introduction}
	
The Maximum Flow Problem is a well studied problem in the area of network flow problems. Given a graph $G=(V,A)$ with non-negative edge capacities $u:A\mapsto \mathbb{R}$, a source $s\in V$, a sink $t\in V\setminus\{s\}$ one searches for a $s$-$t$-flow $f\colon A\mapsto \mathbb{R}_{\geq0}$ such that $0\leq f \leq u$ (capacity constraints), for all $v\neq s,t$ we have $f(\delta^+(v)) - f(\delta^-(v)) = 0$ (flow conservation) and such that the total amount of flow reaching the sink $\text{val}(f)\coloneqq f(\delta^-(t)) - f(\delta^+(t))$ is maximized. Like in standard notation from the literature, we denote by~$\delta^-(v)$ for a node~$v$ the set of ingoing edges, by~$\delta^+(v)$ the set of outgoing edges, and for $S\subseteq A$ abbreviate $f(S)\coloneqq \sum_{a\in S} f(a)$.

In this paper, we study a variant of the family of equal flow problems, which we call the Almost Equal Flow Problems (AEFP). In addition to the data for the Maximum Flow Problem one is given (not necessarily disjoint) homologous subsets $R_i\subseteq E$ for $i=1,\dotsc, k$, monotonically increasing functions~$\Delta_i$ and one requires for the flow~$f$ the \emph{homologous edge set condition} that $f(a)\in[f_i, \Delta_i(f_i)]$ for all $e \in R_i$, $i=1,\dotsc, k$, where $f_i:=\min_{e\in R_i} f(e)$ denotes the smallest flow value of an edge in~$R_i$. In the special case that all $\Delta_i$ are the identity, all edges in a homologous set are required to have the same flow value. This problem is known as the Equal Maximum Flow Problem (EMFP).

The EMFP and related problems have been studied for quite a time. Ali et al.~\cite{ali1988equal} considered a variant of the minimum cost flow problem, where $K$ pairs of edges are required to have the same flow value, which they called \textit{equal flow problem}. An integer version of this problem, where flow on the edges required to be integer, was also studied by Ali et al.~\cite{ali1988equal} and was shown to be $\NP$-complete. Further, they obtained a heuristic algorithm based on a Lagrangian relaxation technique. Meyer and Schulz~\cite{meyers2009integer} showed that the integer equal flow problem is not approximable in polynomial time (unless $\P=\NP$), even if the edge sets are of size two. Ahuja et al.~\cite{ahuja1999algorithms} considered the \textit{simple equal flow problem}, where the flow value on edges of a single subset of the edge set has to be equal. Using Megiddo's parametric search technique \cite{megiddo1978combinatorial,megiddo1981applying}, they present a strongly polynomial algorithm which has a running time of $\O(\{m(m+n\log n)\log n\}^2)$.

Here we provide the first complexity results for the AEMFP. Our complexity and approximation results for the AEMFP are covered in Table~\ref{table:results}, where the first three rows correspond to the variants of the AEMFP.          
\begin{table}[htbp!]
	\centering
	\begin{tabular}{lccll}
		\hline
		Function~$\Delta$ & fractional  & integer & fixed~$k$ & lower bound\\
		&&&& for approximation \\
		\hline
		AEMFP&&&&\\
		\hline
		const.\ deviation& $\P$ & $\NP$ & $\P$ & 2 - $\epsilon$ \\
		concave & $\P$ & $\NP$ & $\P$ & no constant \\ 
		convex & $\NP$ & $\NP$ & $\NP$ & no constant \\ 
		\hline
	\end{tabular}
	\vspace{3pt}
	\caption{Overview of the results for the AEMFP.}\label{table:results}
\end{table}

The columns two to four denote the complexity classes of the different problem variants while the entries of the fifth column contain an upper bound for the best approximation factor for a polynomial algorithm (unless $\P = \NP$). If a function~$\Delta$ is of the form $x\mapsto x+c$ for a fixed constant~$c\geq 0$ we call $\Delta$ a \emph{constant deviation function}. For the AEFMP with $k$~homologous edge sets and constant deviation functions, we obtain a running time of $\O(n^km^k\log(\log(n))^kT_{mf}(n,n+m))$ where $T_{mf}(n,m)$ denotes the running time of a maximum flow algorithm on a graph $G$ with $n$~nodes and $m$~edges. Note that general polynomial time solvability of the AEMFP in case of constant deviation functions also follows from Tardos' Algorithm, see e.g.~\cite{tardos1986strongly}. Our main algorithmic contribution is a combinatorial method which not only works in the constant deviation case but also for concave functions.

The rest of the paper is organized as follows. In Section~\ref{sec: aemfp prob def} we state the formal definition of the AEMFP. The main complexity and approximation results for the general case are provided in Section~\ref{sec: aemfp compl}. The case of a constant deviation function is discussed in Section~\ref{sec: aemfp const} where also the strongly polynomial algorithm based on the parametric search technique is presented. In Section~\ref{sec: aemfp variants} problem variants of the AEMFP are discussed, i.e. the cases of concave and convex deviation functions. We then conclude with a short outlook.

\section{Problem Definition}\label{sec: aemfp prob def}
In this section, we give a formal definition of the \emph{Almost Equal Maximum Flow Problem}. The AEMFP can be formulated as the following optimization problem in the variables $f_e$ $(e\in E)$:
\begin{align}
\textbf{(AEMFP)} \qquad \max \quad& f(\delta^+(s)) - f(\delta^-(s)) \\
\text{s.t.} \quad & f(\delta^+(s)) - f(\delta^-(s)) \geq 0 \\
& f(\delta^+(t)) - f(\delta^-(t)) \leq 0 \\
&  f(\delta^+(v)) - f(\delta^-(v)) = 0  &&\forall v\in V\backslash\{s,t\}\\
&  0 \leq f_r \leq u_r  &&\forall r\in E\\
& f_i \leq f_{r} \leq \Delta_i(f_i)  &&\forall r \in R_{\Delta_i}, \forall  R_{\Delta_i} \label{eq:2},
\end{align}
where $f_i$ denotes the minimum flow value on edges from $R_{\Delta_i}$. In the integral version, we additionally require $f$ to attain only integral values. Note that, in general the above problem is nonlinear due to the nonlinearity of the deviation functions~$\Delta_i$ and condition~\eqref{eq:2}. However, if each~$\Delta_i$ is a constant deviation, then~\eqref{eq:2} becomes $f_i\leq f(r_i)\leq f_i+c_i$ and the AEMFP can be formulated as a linear program.

The simple AEMFP is defined as the AEMFP with just one homologous edge set $R_{\Delta}$. Note that by subdividing edges that are contained in several homologous edge sets, we can assume without loss of generality that the homologous edge sets are disjoint.

\section{Complexity and Approximation}\label{sec: aemfp compl}
In this section, we provide complexity and approximation results for the constant deviation, concave and convex AEMFP.
\begin{theorem}\label{thm: integer constant deviation complexity}
	The integer AEMFP is $\NP$-complete, even if all deviation functions are the same constant deviation function, the homologous sets are disjoint, the capacities are integral, and the graph is bipartite.
\end{theorem}
\begin{proof}
	We prove this by a reduction from \emph{Exact-3-Set-Cover} (X3C). Given an instance of X3C, we construct a graph~$G$ in the following way. For each of the $q$ sets~$S_i\in\mathcal{S}$ we add a node~$S_i$ and for each of the $q$ elements~$a_j\in\mathcal{A}$, we add a node~$a_j$ to $G$. Further, we add a source node~$s$ and a sink node~$t$. We add edges $(s,S_i)$ for $i=1,\dots,k$ with capacity~5, edges $(S_i,t)$ for $i=1,\dots,k$ with capacity~2, edges~$(a_j,t)$ for $j=1,\dots, q$ with capacity~1 and edges between $S_i$ and $a_j$ if $a_j$ is contained in $S_i$ with capacity~1. The edges of the form~$(S_i, t)$ are referred to as \emph{bonus edges}. We define homologous edge sets~$R_i$ as $\{(S_i,a_j): a_j\in S_i\} \cup \{(S_i,t)\}$ for $i=1,\dots,q$ and $R_0 \coloneqq \{(a_j,t): j=1,\dots,q\}$ where all these sets have the same constant deviation function~$\Delta: x\mapsto x+1$. 
	
	Now we want to show that X3C has a solution if and only if there is an integer almost equal maximum $s$-$t$-flow in $G$ with value $\frac{7q}{3}$.	Assume first that X3C has a solution~$S'$. Then, we define an integer almost equal flow as follows:
	\begin{enumerate}[-]
		\item $f(s,S_i) = \begin{cases} 5, & \text{ if } S_i\in S' \\
		1, & \text{ else. }
		\end{cases}$ \\
		\item $f(S_i,a_j) = \begin{cases} 1, & \text{ if } S_i\in S' \text{ and } a_j\in S_i\\
		0, & \text{ else. }
		\end{cases}$ \\
		\item $f(S_i,t) = \begin{cases} 2, & \text{ if } S_i\in S' \\
		1, & \text{ else. }
		\end{cases}$ \\
		\item $f(a_j,t) = 1, \text{ for } i=1,\dots,q$.
	\end{enumerate}
	By definition of the flow~$f$, the homologous edge set constraints, the flow conservation and capacity constraints are fulfilled. Hence, $f$ is an integer almost equal $s$-$t$-flow with flow value~$\frac{7q}{3}$. Assume there is a flow~$f'$ which has greater value than $f$. Due to capacity constraints, this flow must send at least one more unit of flow along an edge of the form $(S_i,t)$ with $S_i\notin S'$. Thus, at least one of the constraints~(\ref{eq:2}) is violated. This is a contradiction to $f'$ being feasible and, hence, $f$ is maximal.
	
	Conversely, assume that the almost equal maximum flow in~$G$ has flow value~$\frac{7q}{3}$. Due to constraint~(\ref{eq:2}), only using bonus edges yield in a flow with value~$q$. Since all flow must be integral and flow preservation holds, we know that $f(S_i,a_j) \in\{0,1\}$ for $i,j=1,\dots,q$. For a fixed node~$S_i$ we distinguish two cases: 
	\begin{enumerate}[-]
		\item If $f(S_i,a_j)=1$ for three such edges, then the bonus edge $(S_i,t)$ carries~$\{0,1,2\}$ units of flow, or
		\item if $f(S_i,a_j)=0$ for at least one of the three edges~$(S_i,a_j)$, then the flow value~$f(S_i,t)$ lies in $\{0,1\}$.
	\end{enumerate}
	Suppose $f(a_j,t)=0$ for at least one edge. By flow preservation, also $f(S_i,a_j)=0$ and with (\ref{eq:2}), we get $f(S_i,t)\leq 1$. With the considerations above, we get an upper bound on the maximum flow value. From edges of the form $(a_j,t)$ we get at most $q-1$ units of flow in total, while at most $\frac{q-3}{3}$ bonus edges can carry $2$ units of flow and $q-\frac{q-3}{3}$ bonus edges carry $1$ unit of flow. Hence, we obtain a maximum flow value of $\frac{7q}{3}-2$.
	
	In order to get the desired flow value of $\frac{7q}{3}$, we need $f(a_j,t)=1$ for all $j=1,\dots,q$. Thus, each $a_j$ receives one unit of flow from some set node~$S_i$. Further, we need $f(S_i,t)=2$ for at least $\frac{q}{3}$ edges. We denote the corresponding indices as $i_1,\dots,i_{\frac{q}{3}}$. This can only happen is for each of these $i_l$ case 1 is true. Now consider $V_l\coloneqq\{a_j : f(S_{i_l}, a_j)=1\}$, $l=1,\dots,\frac{q}{3}$. All these sets are subsets of $\mathcal{A}$ and are pairwise disjoint. Since $|V_l| = 3$, $l=1,\dots,\frac{q}{3}$, we get that \begin{align}\cup_{l=1}^{q/3} V_l = \mathcal{A}.\end{align}
	Hence, $f(S_i,a_j)=0$ for all other $i$. Choosing $S'\coloneqq\{S_{i_l}: l=1,\dots,\frac{q}{3}\}$ gives a X3C solution since each $a_j$ appears exactly once in it. This settles the claim. 
\end{proof}
\begin{theorem}\label{thm: integer constant deviation no 2 approx}     
	Unless $\P=\NP$, for any $\epsilon>0$, there is no polynomial time $(2-\epsilon)$-approximation algorithm for the integer AEMFP, even if we consider disjoint sets and a constant deviation $x\mapsto x+ 1$.
\end{theorem}
\begin{proof}
	We extend the instance of the proof of Theorem~\ref{thm: integer constant deviation complexity} by adding two additional nodes~$t'$, $t''$. Further, we add one edge $(t,t')$ with capacity $\frac{7q}{3}$, $\frac{7q}{3}$ parallel edges $(t',t'')$ with capacity~$1$ and $k$ parallel edges $(s,t'')$ with capacity $2$, which we refer to as \emph{bonus edges}.
	
	For the edges~$(t',t'')$ and $(s,t'')$ a homologous edge set~$R_b$ with $\Delta_b = 1$ is added. The node~$t''$ is the new sink, i.e. we are asking for a $s$-$t''$-flow.
	
	If there exists a solution of X3C, then flow value is equal to $\frac{7q}{3}$, as proven before, and all of the edges $(t,t')$ can be fully saturated. This means, on every bonus edge $(s,t'')$ two units of flow can be send. Overall, this yields in a flow value of
	\begin{align}
	val(f_{yes}) = \frac{7q}{3} + k\cdot\frac{7q}{3}\cdot 2.
	\end{align}
	Now assume that there exists no solution of X3C. Then the maximum flow value is at most $\frac{7q}{3}-1$. Hence, at least one of the edges $(t',t'')$ carries no flow. But since all of the bonus edges are in the same homologous set together with the parallel edges $(t',t'')$, each bonus edge can carry at most one unit of flow. Again, overall we get a flow value of
	\begin{align}
	val(f_{no}) \leq \frac{7q}{3} - 1 + k\cdot \frac{7q}{3}.
	\end{align}
	Thus, for $k\mapsto \infty$, the approximation factor goes to $2$.
\end{proof}	
\section{The Constant Deviation Case}\label{sec: aemfp const}
We start with the simple AEMFP. Let $G=(V,E)$ be a graph with a single homologous edge set~$R$ and constant deviation function $\Delta_R\colon x\mapsto x+ c$. For easier notation, we define $Q:=E\backslash R$ as the set of all edges that are not contained in the homologous edge set~$R$. By the homologous edge set condition~\eqref{eq:2}, we know that the flow value on each of the corresponding edges must lie in an interval $[\lambda^*, \Delta(\lambda^*)] = [\lambda^*, \lambda^* + c]$, where $\lambda^*$ is unknown. For a guess value~$\lambda$ consider the modified network~$G_\lambda$, where we set the upper capacity of every edge in~$R$ to $\lambda+c$ and its lower capacity from~$0$ to~$\lambda$. All edges 
in~$Q$ keep their upper capacities and have lower capacity of~$0$. By $f_\lambda$ we denote a traditional $s$-$t$-flow which is feasible in~$G_\lambda$. 

For an $(s,t)$-cut $(S,T)$ let us denote by
\begin{align*}
g_S(\lambda):=u(\delta^+(S\cap Q)) + \sum_{r\in\delta^+(S\cap R)} \min\{u(r),
\Delta_R(\lambda)\} - \sum_{r\in\delta^-(S\cap R)} \lambda
\end{align*}
its capacity in~$G_\lambda$. By the \emph{Max-Flow
	Min-Cut Theorem} we get
\begin{align*}
\max_{f_{\lambda}} \text{val}(f_{\lambda}) &=
\min_{\text{$(S,T)$ is a $(s,t)$-cut}} g_S(\lambda)
\end{align*}
We summarize some structural results in the following observation.

\begin{observation}\label{strucutralresults}
	The following statements are true.
	\begin{enumerate}[i)]
		\item The function 
		\begin{align*}
		F(\lambda) \coloneqq \min_{\text{$(S,T)$ is a $(s,t)$-cut}} g_S(\lambda) 
		\end{align*}
		is a piecewise linear concave function.
		\item AEMFP can be solved by solving
		\begin{align*}
		\max \left\{ \,F(\lambda): 0\leq \lambda\leq \min_{r\in
			R_\Delta} u(r)\,\right\}.
		\end{align*}
		\item The function~$F(\lambda)$ has at most $2m$
		breakpoints.
		\item The minimum distance between two of these breakpoints is $\frac{1}{m^2}$.
	\end{enumerate}
\end{observation}
\begin{proof}
	\begin{enumerate}[i)]
		\item The function~$F(\lambda)$ is the minimum of linear functions in $\lambda$ and, hence, a concave linear function in $\lambda$.
		\item This is a direct consequence of the Max-Flow Min-Cut Theorem.  
		\item Let $d_R(S) := |\delta^+(S\cap R)| - |\delta^-(S\cap R)|$ denote the number of outgoing and ingoing edges of $R$ in the cut $(S,T)$. A breakpoint of $F(\lambda)$ occurs whenever the cut $(S,T)$ changes in a way that changes $d_R(S)$. As $F(\lambda)$ is concave and $d_R(S)$ counts edges, this can happen at most $2m$ times.
		\item  At a breakpoint, we have 
		\begin{align*}
		u(\delta^+(S\cap Q)) - l(\delta^-(S\cap Q)) + \lambda d_R(S) = u(\delta^+(S'\cap Q)) - l(\delta^-(S'\cap Q)) + \lambda d_R(S')
		\end{align*}
		for two cuts $(S,T)$, $(S',T')$. This gives an expression for $\lambda$ as
		\begin{align*}
		\lambda = \frac{(u(\delta^+(S\cap Q)) - l(\delta^-(S\cap Q)))-(u(\delta^+(S'\cap Q)) - l(\delta^-(S'\cap Q)))}{d_R(S') - d_R(S)}.
		\end{align*}
		Note that the denominator is not zero since $d_R(S)\neq d_R(S')$ by definition of a breakpoint. Therefore, the expression for $\lambda$ is well-defined. Further, we also know $(u(\delta^+(S\cap Q)) - l(\delta^-(S\cap Q)))\neq (u(\delta^+(S'\cap Q)) - l(\delta^-(S'\cap Q)))$. 
		By denoting $U:= \max\{u(r): r\in A\}$ and $L:=\min\{l(r): r\in A\}$, we get
		\begin{align*}
		&\hspace{4pt} (u(\delta^+(S\cap Q)) - l(\delta^-(S\cap Q))) - (u(\delta^+(S'\cap Q)) - l(\delta^-(S'\cap Q)))\\
		=& \hspace{4pt}(u(\delta^+(S\cap Q)) - u(\delta^+(S'\cap Q))) - (l(\delta^-(S\cap Q)) + l(\delta^-(S'\cap Q))) \\
		\leq&\hspace{4pt} mU - mL \\
		=& \hspace{4pt}m(U-L)\\
		\leq& \hspace{4pt}mU
		\end{align*}
		W.l.o.g. we can rearrange the two cuts such that both nominator and denominator are positive, i.e. the nominator lies in $\{1,\dots,mU-mL\}$ and the denominator lies in $\{1,\dots,|R|\}$ since it just counts the edges. Thus, we get for the breakpoint $\lambda$:
		\begin{align*}
		\frac{1}{m} \leq \frac{1}{|R|}\leq \lambda \leq \frac{m(U-L)}{1}
		\end{align*}
		Hence, the smallest distance between two breakpoints is $\frac{1}{m(m-1)} > \frac{1}{m^2}$.
	\end{enumerate}
\end{proof}

Observe that the optimal value $\lambda^*$ is attained at a breakpoint of~$F$. At this point the slope to the left is positive or the slope to the right is negative. If there exists a cut such that the slope is~$0$, we simply take the breakpoint to the left or right of the current value~$\lambda$.

Now we apply the parametric search technique by Megiddo~\cite{megiddo1978combinatorial,megiddo1981applying} to search for the optimal value $\lambda^*$ on the interval~$[0, u_R]$, where $u_R:=\min_{r\in R_\Delta} u(r)$ denotes the minimum upper bound of edges in~$R$. We simulate an appropriate maximum flow algorithm, e.g. the Edmonds-Karp algorithm, for symbolic lower capacities~$\lambda^*$ and upper capacities~$\lambda^* + c$ on the edges in~$R$.

\begin{observation}
	If we run the Edmonds-Karp algorithm, see \cite{edmonds1972theoretical}, to compute a maximum flow with a symbolic input parameter~$\lambda$, all flow values and residual capacities which are calculated during the algorithm steps are of the form~$a + b\lambda$ for $a,b\in\Z$.
\end{observation}
\begin{proof}
	At the start of our algorithm, all flow values are zero. The residual capacities are either integer or of the form $b\lambda$ for some $b\in\Z$, thus can be written as $a+b\lambda$. Whenever we augment flow along a path, we add two values of the form $a+b\lambda$, resulting in a new value of the same form.
\end{proof}

\begin{algorithm}
	\caption{Symbolic Edmonds-Karp}\label{alg: symbolic edmonds karp}
	Input: A graph~$G=(V,E)$, a source~$s$ and a sink~$t$, capacities~$c_{ij}$ for all $(ij)\in E$. \\
	Initialization: Set $f_{ij}\leftarrow 0$ for all $(ij)\in E$. \\
	\While{there exist a path~$p$ in $G_f$}{
		Choose shortest path in $G_f$ w.r.t. the number of edges.\\
		Compute $\Delta\coloneqq\min_{(ij)\in p} c^f_{ij}$ by using Algorithm~\ref{alg: solve comparison} for solving symbolic comparisons.\\
		\ForEach{$e=(i,j)\in p$}{
			$f_{ij}\leftarrow f_{ij} + \Delta$ \\
			$f_{ji} \leftarrow f_{ji} - \Delta$}}
	\Return $f$.	
\end{algorithm}
\begin{algorithm}
	\caption{Solve Comparison}\label{alg: solve comparison}
	Input: A graph~$G$, lower and upper capacity functions~$l,u$, a homologous edge set~$R$ and a test value~$\lambda$\\
	Initialization: Set $\lambda_1 \coloneqq \lambda - \frac{1}{2m^2}$, $\lambda_2 \coloneqq \lambda + \frac{1}{2m^2}$ \\
	\For{$i=1,2$}{
		Compute a maximum flow in $G_{\lambda_i}$. \\
		Compute $d(S_i)$ for the corresponding cuts~$S_i$.}
	\If{$d(S_1) > 0$ and $d(S_2) >0$}{\Return False}
	\Else{
		\If{$d(S_1) < 0$ and $d(S_2) <0$}{\Return True}
		\Else{Compute $\lambda^*$ as the intersection of $g_{S_1}(\lambda)$ and $g_{S_2}(\lambda)$. \\
			\Return $\lambda^*$}
		
	}		
\end{algorithm}

\begin{lemma}\label{lem: running time simple constant deviation}
	Algorithm~\ref{alg: symbolic edmonds karp} computes an almost equal maximum flow in time~$\O(n^3m\cdot T_{MF}(n,n+m))$, where $T_{MF}(n,n+m)$ denotes the time needed to compute a maximum flow on a graph with $n$~nodes and $n+m$~edges. 
\end{lemma}
\begin{proof}
	\textbf{Correctness:} In order to resolve a comparison, we need to decide if the current $\lambda$ is to the left or to the right of the optimal $\lambda^*$. Since $\lambda$ might be a breakpoint, we instead check $\lambda_1 := \lambda -\frac{1}{2m^2}$ and $\lambda_2 := \lambda +\frac{1}{2m^2}$ and denote the corresponding cuts as $S_1, S_2$. If $d_R(S_i)$ is positive for $i=1, 2$, then $\lambda^* > \lambda$, since the flow value increases to the right of the current $\lambda$. In the same way, if $d_R(S_i)$ is negative for $i=1, 2$, then the flow value increases to the left of the current $\lambda$, i.e. $\lambda^*<\lambda$. Hence, since the objective function is concave, the only remaining case is $d_R(S_1) >0$ and $d_R(S_2)<0$. In this case, the flow value decreases in both direction, and we just have to find the unique breakpoint in the interval between $\lambda_1$ and $\lambda_2$. This can be done by computing the intersection of $g_{S_1}(\lambda_1)$ and $g_{S_2}(\lambda_2)$. In total, the comparison is correctly resolved. The proposed algorithm is the usual Edmonds-Karp algorithm except for the comparison which has to be made in order to compute the augmenting path. If the comparison is made correctly, this has no influence on the correctness of the Edmonds-Karp algorithm. However, this is the case, since the question 
	\begin{align*}
	a_1 + b_1\lambda < a_2 + b_2\lambda 
	\intertext{is equivalent to}
	\lambda < \frac{a_2-a_1}{b_1-b_2}. 
	\end{align*}
	Therefore the Edmonds-Karp algorithm with updated comparison resolving is also correct. 
	
	Now for the algorithm: We set the lower and upper capacities of all edges in~$R$ to $\lambda$ to ensure that they have the same flow. Then, by applying the Edmonds-Karp algorithm that runs with a symbolic parameter on $G'$, we find a feasible flow. With this, we have a starting flow for the almost equal maximum problem and we can use the Edmonds-Karp algorithm with updated comparison resolving to find an optimal flow since the residual network respects lower capacities.
	
	\textbf{Running time:} The Edmonds-Karp algorithm with updated comparison resolving has at most $\O(nm)$ iterations. In each of these, a shortest $s$-$t$-path~$P$ w.r.t. the number of edges is calculated, which can be done with Breadth-First-Search and therefore needs time $\O(n+m)$. Then the algorithm computes the minimum residual capacity on the edges of~$P$. Since $P$ has at most $n-1$~edges and each of the residual capacities may depend on the parametric value $\lambda^*$, there are at most $\O(n^2)$ comparisons which the algorithm has to resolve. Updating the residual network takes at most $2m$~comparisons of the form $a+b\lambda < u(r)$ and $a+b\lambda > 0$, thus at most $\O(m)$ comparisons have to be resolved.   
	For each comparison, a maximum flow and the corresponding cut have to be computed. Since all other operations are done in constant time, the running time of resolving one comparison is $\O(T_{MF}(n,m))$.
	Altogether, the Edmonds-Karp algorithm with updated comparison resolving runs in time $\O(nm\cdot(n^2+m)T_\text{Comparison}) \subseteq \O(n^3m\cdot T_{MF}(n,m))$. 
	
	So, in total the algorithm has a running time of $\O(m+(n^3mT_{MF}(n,n+m))+(n^3mT_{MF}(n,m))) \subseteq \O(n^3m\cdot T_{MF}(n,n+m))$.
\end{proof}
The number of comparisons can be decreased by exploiting implicit parallelism~\cite{megiddo1981applying}. 

When building the residual network, the algorithm has to solve $l(r)<f(r)$ and $f(r)<u(r)$ for every edge $r\in A$. Since $f(r) = a+b\lambda$, we have to solve up to $2m$ comparisons. Instead of this, we can first calculate all the values~$v$ for which we want to test $\lambda^* < v$ and sort them. This takes time $\O(m\log m)$ and afterwards we apply a binary search over these values. In total, we can compute the residual network in time $\O(m\log m \cdot T_\text{Comparison})$. With the same trick, the time needed to find the minimum residual capacity on a path~$P$ is $\O(n\log n \cdot T_\text{Comparison})$. This results in a running time of $\O(nm(n\log n + m\log m)T_{MF}(n,n+m))$.

To solve the integer version of the maximum AEMFP, we simply use the optimal value~$\lambda^*$ of the non-integer version and compute two maximum flows on the graphs~$G_{\lfloor\lambda^*\rfloor}$ and $G_{\lceil\lambda^*\rceil}$. By taking the $\text{argmax}\{val(f_{\lfloor\lambda^*\rfloor})$,  $val(f_{\lceil\lambda^*\rceil})\}$ we get the optimal parameter~$\lambda^*_{int}$ for the integer version.

In the general constant deviation AEMFP we consider more than one homologous edge set. By iteratively using the algorithm for the simple constant deviation AEMFP, we obtain a combinatorial algorithm for the general constant deviation AEMFP. We present the algorithm for the case of two homologous edge sets, but it can be generalized to an arbitrary number of homologous edge sets. The idea behind the algorithm is to fix some~$\lambda_1$ and then use the algorithm for the simple case to find the optimal corresponding~$\lambda_2$. Once we found $\lambda_2^*(\lambda_1)$, we check if $\lambda_1$ is to the left, right or equal to $\lambda_1^*$.  Note that the objective function is still a concave function in~$\lambda_1$ and~$\lambda_2$ since it is the sum of concave functions. Also, like in the simple case, all flow values and capacities both in the network $G$ and the residual network~$G_f$ during the algorithm are of the form
\begin{align*}
a + b\lambda_1 + c\lambda_2.
\end{align*}
Note that the running time of the algorithm for the general constant deviation AEMFP increases for every additional homologous edge set roughly by a factor of the running time of the algorithm for the simple constant deviation AEMFP. The next theorem summarizes the results above.

\begin{theorem}\label{thm: running time constant deviation k sets}
	Let $T_{mf}(n,m)$ denote the running time of a not specified maximum flow algorithm on a graph $G$ with $n$~nodes and
	$m$~edges. The AEMFP with~$k$ homologous sets can be solved in time
	\begin{align*}
	\O\left(n^km^k\log(\log(n))^k\cdot T_{mf}(n,n+m)\right)
	\end{align*}
	when we use the Edmonds-Karp algorithm as the underlying maximum flow algorithm. 
\end{theorem}

Note that the running time for an arbitrary number of homologous edge sets becomes exponential. Interestingly, using one of the known faster maximum flow algorithms instead of the Edmonds-Karp algorithm does not seem to yield an improved running time, since using an algorithm based on a push-relabel-technique yields a running time of $\O(n^3k\cdot T_{mf}(n,n+m))$ (see \cite{haese2019aef}).

\section{Problem Variants}\label{sec: aemfp variants}
In the section above we considered the case of a constant deviation of the flow value on edges within a homologous edge set. Now we allow the deviation function to be given as either a convex or a concave function.  
\subsection{The Convex Deviation Case}
If the deviation function is a convex function~$\Delta_{conv}\colon R \mapsto \R_{\geq 0}$, we get
the convex AEMFP. Note that this problem is neither a convex nor a concave program due to the constraint~\eqref{eq:2}. Hence, standard methods of convex optimization can not be applied.  In fact, the next theorem states that, unless $\P=\NP$, one cannot hope to find a polynomial time algorithm that solves the fractional variant of this problem:

\begin{theorem}\label{thm: convex deviation complexity}
	The AEMFP with a convex deviation function~$\Delta$ is $\NP$-complete, even if all deviation functions
	are given as $\Delta_R(x) = 2x^2 +1$ for all homologous sets~$R$, the homologous sets are disjoint, the capacities are integral, and the graph is bipartite.
\end{theorem}
\begin{proof}
	Again we use a reduction from Exact-3-Cover. Given an X3C instance, we construct a network graph in the same way as in the proof of Theorem~\ref{thm: integer constant deviation complexity}. We now show that there exists an almost equal maximum flow with convex deviation functions and flow value $\frac{8q}{3}$ if and only if there exists a solution of X3C.
	
	Assume first that X3C has a solution $S'$. Then, we define an almost equal maximum flow in the same way as in the proof of Theorem \ref{thm: integer constant deviation complexity}.
	Suppose there is an almost equal maximum flow $x$ with flow value $val(x) > \frac{8q}{3}$. Since the capacity of edges $(a_j, t)$ is $1$, the summarized amount of flow on these edges is at most $q$. That means, the sum of flow on edges $(S_i,t)$ has to be greater than $\frac{5q}{3}+1$. Suppose a flow $x'$ uses more than $q$ edges of $(s,S_i)$. Then there must be at least two edges $(S_{i_1}, t)$, $(S_{i_2}, t)$ with $x'(S_{i_k},t)\in (1,3)$ for $k=1,2$. By the homologous edge constraint we know that at least one edge of the form $(S_{i_1}, a_{j})$ and one of the form $(S_{i_2}, a_{j})$ are not fully saturated, i.e. $x'(S_{i_k}, a_j) \in (0,1)$ for $k=1,2$. W.l.o.g. let $x'_1:= x'(S_{i_1}, a_j) \geq x'(S_{i_2},a_j)=: x'_2$. We show that this flow cannot yield the highest possible flow value, since shifting $\epsilon$ between these two edges yields in a higher flow value. Note that the sum $\sum_{k=1}^2 x'_k$ remains the same, only the value of the edges $(S_{i_k}, t)$ changes. In the following, we distinguish the following three cases.
	
	\textbf{Case 1:} The edges $(S_{i_k}, a_j)$ have both the strict smallest flow values in their homologous edge set. Then we shift an $\epsilon$ from $(S_{i_2}, a_j)$ to $(S_{i_1}, a_j)$. This yields
	\begin{alignat*}{3}
	2(x'_1 + \epsilon)^2 + 1 + 2(x'_2-\epsilon)^2 +1 & = 2(x'_1)^2 + 2\epsilon(x'_1 - x'_2) + 2\epsilon^2 + 2 + 2(x'_2)^2 \\
	& > 2(x'_1)^2 + 1 + 2(x'_2)^2 + 1.
	\end{alignat*}
	Hence, the sum of the flow value on the edges $(S_{i_k}, t)$ for $k=1,2$ is higher which is a contradiction to the maximality of $x'$. 
	
	\textbf{Case 2:} $x_1'$ is the smallest value among the flow value on edges of the corresponding homologous edge set and $x_2'$ is strictly larger than the smallest flow value of edges of the corresponding homologous edge set. Then, increasing $x_1'$ by $\epsilon \leq \min\{u_1 - x'_1, x_2' - x'_{2,\min}\}$ with $x'_{2,\min} = \min f_{R_2}$ yields
	\begin{alignat*}{3}
	2(x'_1 + \epsilon)^2 + 1 + 2(x'_{2,\min})^2 +1 & = 2(x_1')^2 + 1 + 2\epsilon x_1' + 2\epsilon^2 + 2(x'_{2,\min})^2 + 1 \\
	& > 2(x'_1)^2 + 1 + 2(x'_{2,\min})^2 + 1.
	\end{alignat*}
	Again, this is a contradiction to the maximality of $x'$.
	
	\textbf{Case 3:} $x'_2$ is the smallest flow value on edges of the related homologous edge set, $x_1'$ is strictly larger than the minimum. Then, shifting $\epsilon$ units of flow from $x'_1$ to $x_2'$ gives us
	\begin{alignat*}{3}
	2(x'_{1,\min})^2 + 1 + 2(x'_2 + \epsilon)^2 +1 & = 2(x_{1,\min}')^2 + 1 + 2(x'_2)^2 + 2\epsilon x_2' + 2\epsilon^2 + 1 \\
	& > 2(x'_{1,\min})^2 + 1 + 2(x'_2)^2 + 1.
	\end{alignat*}
	Also in this case the increasing resp. decreasing by $\epsilon$ yields in a higher flow value --- a contradiction to the maximality of $x'$.
	
	Conversely, assume that the almost equal maximum flow in $G$ has flow value $\frac{8q}{3}$. We need to show that this induces a solution to X3C and this can be done similar to the proof of Theorem \ref{thm: integer constant deviation complexity} under consideration of the three cases above. This settles the proof.
\end{proof}

\begin{theorem}\label{thm: convex deviation approx}
	Unless $\P=\NP$, there is no polynomial time constant factor approximation algorithm for the integer
	convex AEMFP.
\end{theorem}
\begin{proof}
	For this, we use the same reduction as in Theorem~\ref{thm: integer constant deviation no 2 approx}. For the concave deviation function of the set~$R_b$, we choose $\Delta_B^k: x \mapsto kx$. Then, one can see that the maximum flow value is $0$ if no solution of X3C exists and $\frac{7}{3}q + k$ if one exists. Thus, unless $\P=\NP$, no polynomial time constant factor approximation algorithm can exist.
\end{proof}

\subsection{The Concave Deviation Case}    
In contrast to the convex case, which is $\NP$-complete even for the fractional case, the concave case is polynomially solvable since in this case the AEMFP becomes a concave program.  

In the following we describe an algorithm for this variant using again the parametric search technique \cite{megiddo1978combinatorial,megiddo1981applying} and a refinement by Toledo \cite{toledo1993maximizing}. We restrict ourselves to the case of one homologous edge set~$R$, but the algorithm can be extended to an arbitrarily number of homologous edge sets according to~\cite{toledo1993maximizing}.  
As we have seen before, we can solve the AEFMP for fixed lower bounds $\lambda_i$ for each homologous edge set~$R_i$ by a maximum flow computation. Therefore, one can use the parametric search technique by Megiddo~\cite{megiddo1978combinatorial,megiddo1981applying} with symbolic input parameters $\lambda_i^*$ in order to find the (unknown) minimizer of $F$. Also we know that $F(\cdot)$ has no jumps between two breakpoints. Therefore, we restrict ourselves to an interval between two breakpoints and to find a maximizer $x^*_I$ for every such interval~$I$. In a second step, we evaluate all these local maximizers and find the global solution~$x^*$. The problem of finding a maximizer in each of the $m$~intervals can then be done simultaneously. 

By using a standard trick in network optimization, we can assume that our graph $G$ has lower bounds $0$ and we can apply the Edmonds-Karp algorithm to it. The maximum flow algorithm has to answer questions of the form $f_{\lambda^*}(e) = p_e(\lambda^*) \leq^? \min\{c(e), \Delta(\lambda^*)\}$. Such a comparison made by the algorithm is equivalent to the question which sign a polynomial concave function~$p$ at a given point~$x$ has.
\begin{observation}
	During the algorithm, all flow values and all residual capacities can be described by a polynomial~$p$ in $\lambda$ which is of degree at most $q$. 
\end{observation}
\begin{proof}
	We start with the zero flow. In the first step, the residual capacities are either integer or $\Delta(\lambda)$ for some polynomial~$\Delta$ of bounded degree. Thus, the residual capacities can be written as a polynomial~$p(\lambda)$ of bounded degree.
	
	Whenever the algorithm augments the flow along a path, it adds two values of the form $p_1(\lambda), p_2(\lambda)$, where $p_1, p_2$ are again polynomials with degree at most $\delta$. This results in a flow value of the same form. 
\end{proof}
Since the sign of a polynomial is constant between two roots, it is sufficient to restrict ourselves to the roots $\{r_1,\dots,r_l\}$ of the polynomial~$p$. For every root, we evaluate $F$ and test if its evaluation is equal to $x^*$ or else if it is to its left or right.

We know that we can determine the relative position of a point $x$ to $x^*$ by evaluating $F$ at this point. In the case of a constant deviation function, we did this by computing the slope of $F$ at $x$. Here, instead of relying on the slope, we use the idea by Toledo \cite{toledo1993maximizing}. This process is presented in Figure~\ref{fig: toledo}. Evaluating $F$ at a point $x_1$ is a maximum flow computation in the graph $G_{x_1}$, i.e. the graph where the lower bound on edges of the homologous edge set is set to $x_1$. Now we distinguish two cases, either $x_1$ is to the left or to the right of the maximum $x^*$. First we check if we have already evaluated $F$ at a point $x_0$ with $F(x_1) \geq F(x_0)$. If this is the case, we know that the maximum $x^*$ lies in the direction of $x_0$. 

If we have not found a point with larger value in previous evaluations, we cannot resolve the comparison. The Case~0 of Figure~\ref{fig: toledo} shows this situation. Now, we copy the state of the algorithm and proceed in one copy with the presumption that $x_1$ lies to the left of $x^*$ and in the other copy with the presumption that $x_1$ lies to the right of $x^*$. These two cases are depicted as Case~1 (or Case~2 resp.) in Figure~\ref{fig: toledo}. So, on one side, we calculate a maximum flow for some $x_1<x_0$. If $\text{val}(f(x_1)) > \text{val}(f(x_0))$, we can resolve the comparison from above.
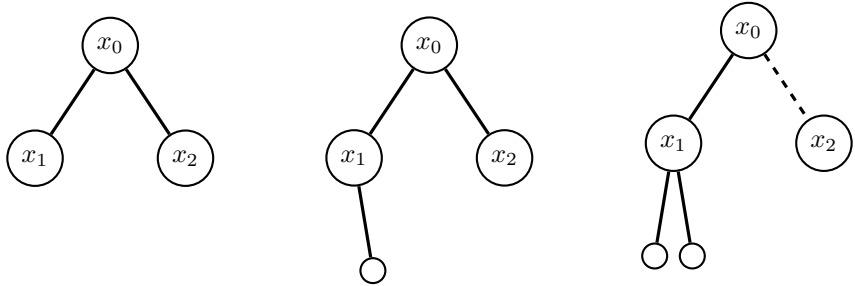
\begin{figure}
	\begin{minipage}{.3\textwidth}
		\begin{tikzpicture}
		\begin{scope}[every node/.style={circle,thick,draw}]
		\node (0) at (0,0.5) {$x_0$};
		\node (1) at (-1,-1) {$x_1$};
		\node (2) at (1,-1) {$x_2$};
		\end{scope}
		\node[text width=4cm] at (0,-4.5) {Case 0: The maximum could be either to the left or to the right of $x_0$.};
		\begin{scope}[
		every node/.style={fill=white,rectangle,sloped},
		every edge/.style={draw=black,very thick}]
		\path [-] (0) edge[color=black] (1);
		\path [-] (0) edge[color=black] (2);
		\end{scope}
		\end{tikzpicture}
	\end{minipage}
	\hfill
	\begin{minipage}{.3\textwidth}
		\begin{tikzpicture}
		\begin{scope}[every node/.style={circle,thick,draw}]
		\node (0) at (0,0.5) {$x_0$};
		\node (1) at (-1,-1) {$x_1$};
		\node (2) at (1,-1) {$x_2$};
		\node[font=\small] (3) at (-0.75,-2.5) {$ $};
		\end{scope}
		\node[text width=4cm] at (0,-4.5) {Case 1: If $F(x_1)>F(x_0)$ holds, only look to the right of $x_1$.};
		\begin{scope}[
		every node/.style={fill=white,rectangle,sloped},
		every edge/.style={draw=black,very thick}]
		\path [-] (0) edge[color=black] (1);
		\path [-] (0) edge[color=black] (2);
		\path [-] (1) edge[color=black] (3);
		\end{scope}
		\end{tikzpicture}
	\end{minipage}
	\hfill
	\begin{minipage}{.3\textwidth}
		\begin{tikzpicture}
		\begin{scope}[every node/.style={circle,thick,draw}]
		\node (0) at (0,0.5) {$x_0$};
		\node (1) at (-1,-1) {$x_1$};
		\node (2) at (1,-1) {$x_2$};
		\node[font=\small] (3) at (-0.75,-2.5) {$ $};
		\node[font=\small] (4) at (-1.25,-2.5) {$ $};
		\end{scope}
		\node[text width=4cm] at (0,-4.5) {Case 2: If $F(x_1)<F(x_0)$ holds, copy the state of the algorithm at $x_1$ and cancel the right side above.};
		\begin{scope}[
		every node/.style={fill=white,rectangle,sloped},
		every edge/.style={draw=black,very thick}]
		\path [-] (0) edge[color=black] (1);
		\path [dashed] (0) edge[color=black] (2);
		\path [-] (1) edge[color=black] (3);
		\path [-] (1) edge[color=black] (4);
		\end{scope}
		\end{tikzpicture}
	\end{minipage}
	\caption{The different possible outcomes of the comparison step.}\label{fig: toledo}
\end{figure}
During the whole process, we only have two copies running at any given time. These two copies can be run in parallel, since they only need to communicate right before the next branching step in order to know which branches of the tree to cut.  

This enables us to prove the following result:
\begin{theorem}
	The AEMFP with a piecewise polynomial concave deviation function~$\Delta$ with maximum degree~$q$ can be solved in polynomial time for one homologous edge set in time $\O(mq\cdot(nm\cdot(n+2m+n^2)(T_{MF}(n,n+m))))$	under the assumption that the roots of a polynomial~$p$ of maximum degree~$q$ can be computed in constant time~$\O(1)$.
\end{theorem}
\begin{proof}
	Since we have to compute the maximum for every interval, we have to run the algorithm $\O(mq)$~times. In each interval, we run the Edmonds-Karp algorithm which has at most $\O(nm)$ iterations. In each iteration, the algorithm needs to find a shortest path $P$ w.r.t. the number of edges, which can be done in $\O(n+m)$ time, for example with a Breadth-First-Search. To find the minimum residual capacity on $P$, the algorithm needs to do $\O(n^2)$ comparisons. For updating the residual network, again $\O(m)$ comparisons are needed.In order to resolve a comparison, first the roots of a polynomial~$p$ of bounded degree are computed, which can be done in constant time~$\O(1)$ by assumption. Evaluating $F$ at a root is a maximum flow computation in the graph~$G'$ where the lower bounds have been eliminated. Since this graph $G'$ has $n$ nodes and $n+m$ edges, we write $T_{MF}(n,n+m)$ for the time needed to compute a maximum flow in this graph. Overall, this yields in a running time of $\O(mq\cdot(nm\cdot(n+2m+n^2)(T_{MF}(n,n+m))))$.
\end{proof}

In the worst case, our algorithm yields a better running time than a direct implementation of the Megiddo-Toledo algorithm for maximizing non-linear concave function in~$k$ dimensions, which runs in $\O((T_{mf}(n,m))^{2^k})$ (\cite{toledo1993maximizing}).

The integral version of the concave AEMFP turns out to be still hard to solve and hard to approximate.

\begin{theorem}\label{thm: concave complexity}
	The concave integer AEMFP is $\NP$-complete.
\end{theorem}
\begin{proof}
	The proof is similar to the proof of Theorem~\ref{thm: integer constant deviation complexity}.
\end{proof}

\begin{theorem}\label{thm: concave approx}
	Moreover, unless $\P=\NP$, there is no polynomial time constant factor approximation algorithm for the integer concave AEMFP.
\end{theorem}
\begin{proof}
	For this, we use the same reduction as in Theorem~\ref{thm: integer constant deviation no 2 approx}. For the concave deviation function of the set~$R_b$, we choose $\Delta_B^k: x \mapsto kx$. Then, one can see that the maximum flow value is $0$ if no solution of X3C exists and $\frac{7}{3}q + k$ otherwise. Thus, unless $\P=\NP$, no polynomial time constant factor approximation algorithm can exist.  
\end{proof}

	\section{Outlook}

 In this paper, we considered a novel class of flow problems, which we call almost equal flow problems. These are to be understood as a generalization of the equal flow problems. The motivation to study these problems comes from its application of finding an optimal load schedule between energy suppliers and energy consumers, where one tries to use flexibility to shift power consumption from peak times to times of lower grid utilization. This can be modeled in a time-expanded graph that shows the flow of electricity to a consumption unit. Such a consumption unit often has technical limitations that prevent the load from changing too much between successive points in time. These types of constraints can now be modeled by the almost equal property in such a time-expanded flow network.

For the Almost Equal Maximum Flow case we proved that the problem of finding such an optimal integer flow turns out to be hard to solve in general, regardless of whether the function is given by an affine transformation, a concave function or a convex function. Further, even finding an optimal maximum fractional flow for a convex deviation function is $\NP$-hard to find. Nevertheless, by using the parametric search technique by Megiddo we provide strongly polynomial algorithms if the number of homologous sets is given by a constant and the deviation function is either a constant or concave deviation function. As variants of the \emph{Almost Equal Flow Problems}, we discussed different deviation functions and obtained complexity results for these problem variants. 

Future research should be directed to the question if the network structure can be exploited in order to get faster algorithms for special graph classes. Furthermore, the obtained results can be extended to the \emph{Almost Equal Minimum Cost Flow Problem} in a similar way. 

\bibliography{aemfp}
\bibliographystyle{alpha}

\vfill
\pagebreak
\small
\vskip2mm plus 1fill
\noindent
Version \today{}
\bigbreak

\noindent
Till Heller\\
Department of Optimization\\
Fraunhofer ITWM, Kaiserslautern\\
Germany\\
ORCiD: 0000-0002-8227-9353\\

Rebekka Haese\\
Sven O. Krumke\\
Optimization Research Group, Department of Mathematics\\
Technische Universit\"at Kaiserslautern, Kaiserslautern\\
Germany\\

\end{document}